\documentclass[12pt]{article}

\usepackage[sort,longnamesfirst]{natbib}

\usepackage{amsbsy,amsmath,amsthm,amssymb,graphicx,verbatim,color}
\usepackage{subfig}

\usepackage{geometry}
\newtheorem{proposition}{Proposition}
\newtheorem{lemma}{Lemma}
\newcommand{\diag}{\text{diag}}

\newcommand{\Real}{ \mathbb{R}}
\newcommand{\integerZ}{ \mathbb{Z}}
\newcommand{\bthet}{ \mbox{\boldmath $\theta$}}

\newcommand{\Poi}{\text{Poi}}

\newcommand{\muhat} {\hat{\mu}}
\newcommand{\bmuhat} {\hat{\bmu}}

\newcommand{\bthethat}{\hat{\bthet}}
\newcommand{\bphihat}{\hat{\bphi}}

\newcommand{\Vinv}{V^{-1}}

\newcommand{\pari}{\hspace{\parindent}}

\newcommand{\bc}{\begin{center}}
\newcommand{\ec}{\end{center}}
\newcommand{\beq}{\begin{equation}}
\newcommand{\eeq}{\end{equation}}
\newcommand{\bea}{\begin{eqnarray}}
\newcommand{\eea}{\end{eqnarray}}
\newcommand{\beas}{\begin{eqnarray*}}
\newcommand{\eeas}{\end{eqnarray*}}

\newcommand{\bi}{\begin{itemize}}
\newcommand{\ei}{\end{itemize}}
\newcommand{\bd}{\begin{description}}
\newcommand{\ed}{\end{description}}
\newcommand{\be}{\begin{enumerate}}
\newcommand{\ee}{\end{enumerate}}
\newcommand{\bv}{\begin{verbatim}}
\newcommand{\ev}{\end{verbatim}}
\newcommand{\bt}{\begin{tabbing}}
\newcommand{\et}{\end{tabbing}}

\newcommand{\btheta}{ \mbox{\boldmath $\theta$}}

\newcommand{\bphi}{ \mbox{\boldmath $\phi$}}
\newcommand{\bmu}{ \mbox{\boldmath $\mu$}}

\newcommand{\bX}{ {\bf X} }
\newcommand{\bY}{ {\bf Y} }

\newcommand{\Thetahat}{\hat{\Theta}}

\newcommand{\half}{\frac{1}{2}}

\DeclareMathOperator{\Diag}{Diag}

\begin{document}

\title{On automating Markov chain Monte Carlo for a class of spatial models}

\author{Murali Haran\\
  School of Statistics \\ The Pennsylvania State University \\ {\tt
    mharan@stat.psu.edu} \and Luke Tierney
  \\Department of Statistics and Actuarial Science \\ University of
  Iowa \\ {\tt luke@stat.uiowa.edu}} \date{} 
\maketitle
 \begin{abstract} 
   Markov chain Monte Carlo (MCMC) algorithms provide a very general
   recipe for estimating properties of complicated distributions.
   While their use has become commonplace and there is a large
   literature on MCMC theory and practice, MCMC users still have to
   contend with several challenges with each implementation of the
   algorithm. These challenges include determining how to construct an
   efficient algorithm, finding reasonable starting values, deciding
   whether the sample-based estimates are accurate, and determining an
   appropriate length (stopping rule) for the Markov chain. We
   describe an approach for resolving these issues in a theoretically
   sound fashion in the context of spatial generalized linear models,
   an important class of models that result in challenging posterior
   distributions.  Our approach combines analytical approximations for
   constructing provably fast mixing MCMC algorithms, and takes
   advantage of recent developments in MCMC theory. We apply our
   methods to real data examples, and find that our MCMC algorithm is
   automated and efficient. Furthermore, since starting values,
   rigorous error estimates and theoretically justified stopping rules
   for the sampling algorithm are all easily obtained for our
   examples, our MCMC-based estimation is practically as easy to
   perform as Monte Carlo estimation based on independent and
   identically distributed draws.
   \\
 \end{abstract} 
\section{Introduction}\label{sec:intro}

Markov chain Monte Carlo (MCMC) methods have become standard tools in
Bayesian inference and other areas in statistics. When inference is
based on some distribution $\pi$, the Metropolis-Hastings algorithm
provides a general recipe for constructing a Markov chain with $\pi$
as its stationary distribution. 
From a careful practitioner's standpoint, however, MCMC-based
inference poses several challenges. In addition to developing and
fine-tuning an MCMC algorithm that produces accurate sample-based
inference quickly, MCMC users also need to determine an appropriate
length for the Markov chain. These issues pose a challenge to
non-experts since, even for a specific class of models, the MCMC
algorithm needs to be tuned carefully to the posterior distribution
resulting from each new data set. Also, the commonly used MCMC
``convergence diagnostics'' used to determine stopping rules for the
algorithm may be unreliable \citep{Cowl:Carl:1996} and may not always
be directly connected to the central goal of MCMC-based inference,
which is to estimate properties of the posterior distribution $\pi$ up
to some desired level of accuracy. Furthermore, users may not always
have a reasonable approach for finding starting values for the
algorithm.

In this paper, we consider approaches for resolving these closely
related issues in a rigorous manner.  We describe how we can construct
provably efficient MCMC algorithms where the MCMC standard errors,
which represent the accuracy of sample-based estimates, can be
estimated consistently.  These MCMC standard errors can, in turn, be
used in a theoretically sound approach to determine an appropriate
length for the Markov chain using new developments in MCMC output
analysis \citep{fleg:hara:jone:2008,Jone:Hara:Caff:Neat:fixe:2006}.
Our approach also automatically provides reasonable starting values
for the MCMC algorithm. Hence, the resulting MCMC-based inference is
automated, theoretically sound, and practical because: (i) initial
values for the chain are obtained automatically, (ii) the algorithm
produces accurate answers quickly and is therefore useful in practice,
(iii) an appropriate Monte Carlo sample size (length for the Markov
chain) is determined automatically, and (iv) the accuracy of the
sample-based estimates can be assessed rigorously.

A key tool we develop for the samplers described here is an accurate
heavy-tailed approximation to the posterior distribution of interest.
In this paper we describe how one can construct such an approximation
and investigate its use in constructing two samplers: a fast mixing
independence Metropolis-Hastings chain \citep{Tier:mark:1994} and the
classic rejection (`accept-reject') sampler.  We study these
algorithms in the context of a popular class of spatial generalized
linear models. Naive MCMC samplers are known to mix poorly for these
models, and little is known about the theoretical properties of the
samplers.  These models are closely related in form to the
geostatistical models described in \cite{Digg:Tawn:Moye:1998}, which
are Bayesian versions of generalized linear models
\citep{McCu:Neld:gene:1999} for spatial data. Hence, our approach, or
variants of it, may be applicable to several other important
statistical models. Our MCMC algorithm is virtually as simple to use
as simple Monte Carlo using independent and identically distributed
draws. An underlying assumption of our work here is that sample-based
inference is of interest, for instance when propogating uncertainties
associated with inferences based on one model or model-component to
other models, or when modelers are interested in reporting estimates
at a level of accuracy that they would like to control. When
sample-based inference is not critical, we note that purely analytical
approximations \cite[cf.][]{Tier:Kada:1986,Rue:Mart:Chop:2009} may, of
course, be completely `automatic' approaches for approximate 
inference since they avoid the MCMC issues outlined above.

In Section \ref{sec:samplers} we provide a general description of the
samplers we consider, along with a discussion of relevant MCMC theory.
In Section \ref{sec:model} we discuss the class of models to which we
apply our methods. In Section \ref{proposal} we describe a general
approach for deriving heavy-tailed approximations in hierarchical
models, discuss how such approximations can be used to construct fast
mixing MCMC algorithms, and provide theoretical details. Section
\ref{application} describes the application of these methods to
several real data sets.  We conclude with a discussion of our results
in Section \ref{discussion}.

\section{Background}\label{sec:samplers}
This section provides a brief overview of the samplers we consider,
along with some relevant theoretical background. In Section
\ref{subsec:MCMC} we provide some basic theory on Markov chain Monte
Carlo, which sets the stage for the fast mixing MCMC sampler. In
Section \ref{subsec:rejection} we briefly review rejection sampling
since we are later able to use our heavy-tailed approximation to
construct effective rejection samplers in some cases.

\subsection{Markov chain Monte Carlo basics}\label{subsec:MCMC}
Our goal is typically to estimate expectations with respect to a
distribution $\pi$. That is, we are interested in $E_{\pi} g(x)
=\int_{\Omega} g(x)\pi(dx)$, where $g$ is a real-valued
$\pi-$integrable function on $\Omega$. Consider a Harris-ergodic
Markov chain $\bX=\{X_1,X_2,\dots\}$ with state space $\Omega$ and
stationary distribution $\pi$ \cite[for definitions
see][]{Meyn:Twee:1993}. If $E_{\pi}| g(x)|<\infty$, we can appeal to
the ergodic theorem $$
\bar{g}_n = \frac{1}{n} \sum_{i=1}^n
g(X_i)\rightarrow E_{\pi} g(x) \mbox{ with probability 1}.$$
Now let
$P^n(x,\cdot)$ be the $n$-step transition kernel for this chain, so
that $P^n(x,A) = P(X_{i+n}\in A \mid X_i = x)$ for $x\in \Omega$ and
any measurable set $A$. Then it also follows that $$
\lim_{n\to
  \infty} \lVert P^n(x,\cdot) - \pi(\cdot) \rVert_{TV} = 0 \:\:\:
\mbox{ for any }x\in \Omega,$$
where the convergence is in total
variation distance, which implies convergence in distribution
\citep{Bill:1999}. However, to estimate standard errors, we need to
appeal to a Central Limit Theorem (CLT), which does not hold in
general. For a CLT to hold, the rate of convergence of the Markov
chain to $\pi$ is critical. A Markov chain is said to be {\it
  geometrically ergodic} if for some constant $t\in (0,1)$ and
$\pi$-almost surely finite function $M: \Omega \to \Real^+$ and $n \in
\mathbb{N}$
\begin{equation*}
  \lVert P^n(x,\cdot) - \pi(\cdot)\rVert_{TV} \leq M(x) \, t^n \mbox{ for any } x \in \Omega.
\end{equation*}
If $M(x)$ is bounded, then $\bX$ is {\it uniformly ergodic}. If $\bX$ is uniformly
ergodic and $E_{\pi} g^{2} < \infty$, we have a CLT
\begin{equation*}
\sqrt{n} (\bar{g}_{n} - E_{\pi} g) \stackrel{d}{\rightarrow}
\text{N} (0, \sigma^{2}_{g})
\end{equation*}
as $n \rightarrow \infty$ with $\sigma^{2}_{g} = \text{var}_{\pi} \{
g(X_{1})\} + 2 \sum_{i=2}^{\infty} \text{cov}_{\pi} \{ g(X_{1}),
g(X_{i})\}$. 
For a review of Markov chain CLTs under other sets of sufficient
conditions see \citet{jone:2004} and \citet{robe:rose:2004}.  For a
general review of MCMC theory refer to \cite{Tier:mark:1994}.  

If the CLT holds as above then \cite{Jone:Hara:Caff:Neat:fixe:2006}
provide a consistent estimate of $\sigma_g^2$, the MCMC standard error
in $\bar{g}_n$, which can then be used to determine how long to run
the Markov chain. In the `fixed width' approach to determining chain
length, the user stops the simulation when the width of a confidence
interval based on an ergodic average is less than a user-specified
value \citep{fleg:hara:jone:2008,Jone:Hara:Caff:Neat:fixe:2006}. Thus,
{\it for a uniformly ergodic chain, sample-based inference is
  comparable to i.i.d.\ Monte Carlo in many ways}: a CLT holds under
similar conditions, a consistent estimate of Monte Carlo standard
errors is easily obtained, and the error estimate can be used to
determine a stopping rule for the sampler.  Note that it is
challenging to construct uniformly ergodic Markov chains for real
problems, and it is generally not easy to prove that a given Markov
chain is uniformly or geometrically ergodic; most such proofs have
relied on establishing drift and minorization conditions
\cite[cf.][]{hobe:geye:1998}.

We briefly describe the consistent batch means approach to calculating
Monte Carlo standard errors. Suppose the Markov chain $\bX$ is run for
a total of $n=ab$ iterations 
(hence $a$ and $b$ are implicit functions of $n$) and define
\begin{equation*}
  \bar{Y}_{j} := \frac{1}{b} \sum_{i=(j-1)b+1}^{jb} g (X_{i})
\hspace*{5mm} \text{ for } j=1,\ldots,a \; .
\end{equation*}
The batch means estimate of $\sigma_{g}^{2}$ is
\begin{equation*}
\hat{\sigma}_{g}^{2} = \frac{b}{a-1} \sum_{j=1}^{a} (\bar{Y}_{j} -
\bar{g}_n)^{2}  \; .
\end{equation*}
\citet{Jone:Hara:Caff:Neat:fixe:2006} showed that if the batch size
and the number of batches are allowed to increase as the overall
length of the simulation increases by setting $b_{n}=\lfloor \sqrt{n}
\rfloor$ and $a_{n} = \lfloor n/ b_{n} \rfloor$ then
$\hat{\sigma}_{g}^{2} \to \sigma_{g}^{2}$ with probability 1 as $n \to
\infty$.  This is called consistent batch means (CBM) to distinguish
it from the standard (fixed number of batches) version.  This is an
attractive approach to estimating standard errors since it is easy to
compute and holds under the regularity conditions that the chain is
uniformly ergodic and $E_{\pi} |g|^{2} < \infty$ (though these are not
the only set of sufficient conditions; see
\citet{Jone:Hara:Caff:Neat:fixe:2006} for details).

\subsection{Rejection sampling}\label{subsec:rejection}

Rejection or `accept-reject sampling' \citep{vonneumann1951vtu} is a
well established, simple but powerful method for generating random
variates from a given distribution $\pi$ with support $\Omega$
\cite[also see][]{robe:case:2005}.
Assume we have a proposal distribution $r$ so that we can draw random
samples from $r$, and we know $B$ such that
\begin{equation}\label{eq:rejcond}
  \mbox{ ess} \sup_{x \in \Omega} \frac{\pi(x)}{r(x)} < B,\mbox{  for some }B<\infty,
\end{equation}
where ``ess sup'', the essential supremum (the supremum over all but a
set of measure zero), is taken with respect to $\pi$.
The accept-reject sampling algorithm is as follows:
\begin{itemize}
\item Draw $X \sim r$ and draw $U \sim $ Uniform(0,1).
\item If $U \leq \frac{\pi(X)}{r(X) B }$ return $X$, else do not return $X$.
\end{itemize}
Values returned by the above algorithm are distributed according to
$\pi$. Note that we only need to know both $\pi$ and $r$ up to a
constant of proportionality, that is, we could replace $\pi(x)$ and
$r(x)$ with unnormalized functions $h(x)$ and $q(x)$ where $h(x)
\propto \pi(x)$ and $q(x) \propto r(x)$. However, in addition to
satisfying (\ref{eq:rejcond}), we also need a specific value of $B$
that satisfies this condition. These are stringent requirements, and
explains why rejection sampling is rarely considered a practical
option for sample-based inference for realistic Bayesian models. 

\section{Generalized linear models for spatial data }\label{sec:model}

\pari Spatial generalized linear models are very convenient models for
spatial data when the sampling mechanism is known to be non-Gaussian.
The spatial dependence can be modeled via Gaussian processes
\citep{Digg:Tawn:Moye:1998} or Gaussian Markov random fields (GMRFs)
\cite[cf.][]{Rue:Held:2005}. 
For brevity, we only describe GMRF-based models for count data as this
is the example used later in this paper.  Other models such as
Gaussian process-based models may be specified in analogous fashion.

Consider the following hierarchical spatial model for areal data (data
arising as sums or averages over geographic regions): the count in
region $i$, denoted by $Y_{i}$ for $i=1,\dots,N$, is modeled as a
Poisson random variable with mean $E_{i}\exp(\mu_i)$. $E_{i}$, assumed
fixed and known, is the count in region $i$ when assuming constant
rates for all regions and is typically the product of the population
of the $i$th region and the overall rate in the entire study region.
$\mu_i$ is the log-relative risk specific to the $i$th region. Hence
the $Y_i$s are modeled as conditionally independent random variables,
\begin{equation}\label{eq:Poislikelihood2}
  Y_{i} \mid \mu_i \sim \Poi (E_{i}e^{\mu_i}),\:\:i=1,\dots,N,
\end{equation}
with $\mu_i$ modeled linearly as $\mu_i = \theta_i + \phi_i$, and
$\btheta=(\theta_1,..,\theta_N)^T$ and $\bphi=(\phi_1,..,\phi_N)^T$
are vectors of random effects. The $\theta_i$'s are independent and
identically distributed normal variables, while the $\phi_i$'s are
assumed to follow a GMRF. In this way, each $\theta_i$ captures the
$i$th region's extra-Poisson variability due to area-wide
heterogeneity, while each $\phi_i$ captures the $i$th region's excess
variability attributable to regional clustering. These distributions
are specified as follows:
\[ 
\theta_i|\tau_h \sim N(0,1/\tau_h) \mbox{, and } \phi_i|\phi_{j \neq
  i} \sim N(\mu_{\phi_i}, \sigma^2_{\phi_i}), i=1,...,N\:,\]
\[ \mbox{where }\mu_{\phi_i} = \frac{\sum_{j \neq i}w_{ij}\phi_j}{\sum_{j \neq i}w_{ij}} \mbox{ and } \sigma^2_{\phi_i}= \frac{1}{\tau_c{\sum_{j \neq i}w_{ij}}}. 
\]
The $\mu_{\phi_i}$ for a region $i$ is thus a weighted average of the
clustering parameters in other regions. Here we use the most common
weighting, where $w_{ij}$ is taken as $1$ if regions $i$ and $j$ are
immediate neighbours, and 0 otherwise.  This improper prior can be
written as
\begin{equation*}
\bphi | \tau_c \propto \tau_c^{(N-1)/2}\exp\left(-\frac{\tau_c}{2}\bphi^TQ\bphi\right)
\end{equation*}
where
$$Q_{ij}=\left\{
\begin{array}{l l}
 \phantom{-}n_i & \mbox{if }i=j\\
 \phantom{-}0 & \mbox{if }i\mbox{ is not adjacent to }j\\
-1 & \mbox{if }i\mbox{ is adjacent to }j\\
\end{array} \right.
$$
The model specification is completed by specifying priors on the
precision (inverse variance) parameters, $\tau_h \sim
$Gamma($\alpha_h,\beta_h$) and $\tau_c \sim
$Gamma($\alpha_c,\beta_c$).  Inference for this model is based on the
$2N +2$ dimensional posterior distribution $\pi(\btheta,\bphi,
\tau_h,\tau_c \mid \bY)$, where $\bY=(Y_1,\dots,Y_n)^T$. The 
posterior distribution is provided in Appendix \ref{app:derivations}.

This model was used by \citet{Besa:York:Moll:1991} for Bayesian image
restoration, and has since become popular in disease mapping
\citep{mollie:1996}.  GMRFs are very widely used in hierarchical Bayes
models \cite[see for example][]{ban:carl:gelf:2004}. The distributions
arising from such models are challenging enough so standard MCMC
samplers for the posterior distributions are known to mix poorly. A
variety of improved MCMC approaches for these models have been studied
in \cite{Knor:Rue:2002} and \citet{Hara:Hodg:Carl:2003} but their
theoretical properties are not known. The issues described in Section
\ref{sec:intro} are hence unresolved: determining an appropriate chain
length is difficult, and MCMC standard errors of estimates are not
estimated rigorously.

\section{Constructing automated samplers}\label{proposal}
We describe an approach for constructing samplers that are virtually
automated for the class of spatial generalized linear models we
consider in this paper. To achieve this, we first derive an
approximation for posterior distributions arising from hierarchical
models with latent Gaussian random fields.  A heavy-tailed version of
this approximation is then used to create a proposal distribution that
allows for the construction of the samplers described in Section
\ref{sec:samplers}.

We derive the approximation by first obtaining an approximation for
the joint distribution in a form that allows for a large number of
parameters to be integrated out analytically.  The method of
integrating out model parameters to sample from a lower dimensional
marginal distribution has been explored by several authors, including
\citet{Wolf:Kass:2000,Game:More:Rue:2003,Ever:Morr:2000,Ever:2001,blinded:2003}.
In each of the hierarchical models considered above, it is possible to
analytically obtain the exact marginal distribution of the variance
components.  Furthermore, any MCMC algorithms used when sampling from
the marginal distributions neither have known theoretical properties,
nor rigorous estimates of error associated with them.  For
hierarchical models in general and the spatial models considered here
in particular, exact formulae are not unavailable for any marginal
distributions.

\subsection{A heavy-tailed approximation for hierarchical models}\label{genmargdist}
We now describe our general approach for deriving an approximation for
a posterior distribution of a variance component model. Denote the
precision parameters by $\Lambda$, the random effects by $\Theta$ and
the data by $\bY$. Suppose the distribution of interest is
$\pi(\Theta, \Lambda|\bY)$ and we can derive an approximation
$\hat{\pi}(\Theta, \Lambda|\bY)$ for which we can analytically obtain
the marginal distribution of the variance components,
$s_1(\Lambda|\bY)$, and the conditional distribution,
$s_2(\Theta|\Lambda,\bY)$. Approaches for deriving such a $\hat{\pi}$
involve using a Gaussian approximation to the likelihood, as described
in Section \ref{subsec:margdistr}; in general, the Laplace
approximations may be useful for this as well
\cite[cf.][]{robe:case:2005}.  Our general approach, first described
in \cite{blinded:2003}, can be summarized as follows:
\begin{itemize}
\item Step 1: Approximate the target posterior, $\pi(\Theta,
  \Lambda|\bY)$, as arising from a generalized linear model, by $\hat{\pi}(\Theta,
  \Lambda|\bY)$.
\item Step 2: Analytically integrate $\hat{\pi}(\Theta, \Lambda|\bY)$
  with respect to $\Theta$ to obtain the approximate marginal
  posterior $s_1(\Lambda|\bY)$. $\hat{\pi}(\Theta, \Lambda|\bY)$ can
  then be written as the product $s_1(\Lambda|\bY)
  s_2(\Theta|\Lambda,\bY)$ where $s_2$ is the easily obtained
  approximate conditional distribution of the random effects.
\item Step 3: Find $r_1$ and $r_2$ such that they are heavier tailed
  distributions with similar shapes to $s_1$ and $s_2$ respectively.
  $r(\Theta,\Lambda|\bY)=r_1(\Lambda|\bY) r_2 (\Theta|\Lambda,\bY)$ is
  then a heavy-tailed approximation to $\pi$. If we can demonstrate
  that $r$ is heavy-tailed with respect to $\pi$, that is, it
  satisfies (\ref{eq:rejcond}), it can be used to construct samplers
  with desirable properties.
\end{itemize}
We believe the approach outlined above can be used to derive useful
approximations for several interesting and important Bayesian models.
We focus our attention here on showing how it can be applied to the
models described in Section \ref{sec:model}.

\subsection{A heavy-tailed approximation for a spatial generalized linear model}\label{subsec:margdistr}
This subsection provides an outline of the derivation of the
approximate marginal distributions ($s_1$) of the precision parameters
for the GMRF model described in Section \ref{sec:model}.  Details are
provided in Appendix \ref{app:derivations}. A Gaussian approximation
for the likelihood (\ref{eq:Poislikelihood2}) is
\begin{equation}\label{eq:Approxlikelihood2}
Y_{i} \sim N(E_{i}e^{\mu_i},E_{i}e^{\mu_i})
\end{equation}
Let $\hat{\mu_i}$ be $\log(Y_i/E_i)$. The delta method gives us the
approximation
\begin{equation}\label{eq:delta2}
  \hat{\mu_i} \sim N(\theta_i+\phi_i, 1/Y_i). 
\end{equation}
Note that if $Y_i$ is 0, we replace it with 0.5 when constructing the
approximation. A similar strategy was described in
\citet{Hara:Hodg:Carl:2003} in the context of block MCMC sampling. If
we denote $\bthet=(\theta_1,\dots,\theta_N)^T$,
$\bphi=(\phi_1,\dots,\phi_N)^T$ and $\Theta=(\bthet^T,\bphi^T)^T$ we
can analytically integrate the approximate joint posterior
distribution $\hat{\pi}(\Theta,\tau_h,\tau_c|\bY)$ with respect to
$\Theta$ to obtain $s_1(\tau_h,\tau_c|\bY)$, an approximation to the
marginal distribution of $(\tau_h,\tau_c)$.

If we let $\Vinv = \Diag(Y_1,\dots,Y_N)$ and $\bmuhat =
(\muhat_1,\dots,\muhat_N)^T$, the approximate distribution of the
random effects parameters, conditional on the precision parameters is
 \begin{equation}\label{eq:SMCMCpost}
   s_2(\Theta | \tau_h,\tau_c,\bY) \sim N\left(C^{-1} \left(-\half D^T\right)\;,\; C^{-1}\right) \;,
 \end{equation}
where 
\begin{equation*}
  C_{2N\times 2N} =
  \begin{bmatrix}
    \Vinv + \tau_h I & +\Vinv\\
    +\Vinv & \Vinv + \tau_c Q
  \end{bmatrix},
\end{equation*}
and
\begin{equation*}
  D_{1 \times 2N}= (-2\bmuhat^TV^{-1}, -2\bmuhat^TV^{-1}).
\end{equation*}
We use heavy-tailed distributions $r_1$ and $r_2$ that have roughly
the same shape and scale as $s_1$ and $s_2$ respectively.  We let
$r_1=r_{1a}r_{1b}$, where $r_{1a}$ and $r_{1b}$ are independent log-t
distributions. Among candidate distributions considered were the
gamma, Weibull and log-normal, but the log-t was found to be the most
appropriate due to its tail behavior and flexibility.  $r_2$ was then
obtained by using a multivariate-t with the mean and covariance of
$s_2(\bthet,\bphi|\tau_h,\tau_c,\bY)$.  We can now state the following
result.
\begin{proposition} \label{prop:envelope} Let
  $r(\tau_h,\tau_c,\Theta)= r_1(\tau_h,\tau_c) r_2(\Theta \mid
  \tau_h,\tau_c)$. Let $\pi$ be the posterior distribution
  corresponding to the Poisson-Markov random field model in Section
  \ref{sec:model}. Then, 
  $$\mbox{ ess } \sup \frac{\pi (\tau_h,\tau_c,\Theta)}{r(\tau_h,\tau_c,\Theta)} < B, $$
for some $B<\infty$, with $\tau_h>0, \tau_c>0, \Theta \in \Real^{2n}.$
\end{proposition}
\begin{proof} 
See Appendix~\ref{app:proof}.
\end{proof}

To specify the parameter values for $r_1$ and $r_2$, we find $r_1$ and
$r_2$ that best match $s_1$ and $s_2$. This can be done in a number of
ways, but a simple and effective approach we used was as follows: The
profiles of the log-transformed function $s_1(\tau_h,\tau_c|\bY)$ along
$\tau_h$ and $\tau_c$ were plotted. The heavy-tailed proposals for the
precision parameters were found on the log-scale by matching the mode
and variance of the t-distributions to the approximate log-scale
profiles.  The corresponding log-t distributions $r_{1a}(\tau_h)$ and
$r_{1b}(\tau_c)$ were used jointly as the proposal
$r_1(\tau_h,\tau_c|\bY)$ for ($\tau_h,\tau_c$). It is easy to draw
$(\tau_h^*,\tau_c^*,\Theta^*)$ from $r_1 r_2$:
\begin{enumerate}
\item Draw $\tau_h^* \sim r_{1a}$ and $\tau_c^*\sim r_{1b}$.
\item Conditional on the values drawn above, simulate $\Theta^* \sim r_2(\Theta|\tau_h^*,\tau_c^*,\bY)$.  
\end{enumerate}

\subsection{Sampling algorithms}\label{subsec:samplingalg}
We now provide some details of the samplers constructed based on the
approximations obtained in the previous sections.  The ease of
producing samples from $r$, combined with Proposition 1, allows us to
construct two samplers:
\begin{enumerate}
\item Rejection sampler: $r$ can be used as a proposal distribution in a classical
  rejection sampling algorithm.
\item Independence Metropolis-Hastings: $r$ can be used as a proposal
  in an independence chain algorithm, where every proposed update to
  the Markov chain is obtained from $r$, regardless of the current
  value of the chain \citep{Tier:mark:1994}. In particular, we
  consider an independence chain where the entire state space is
  updated in a single block. The starting value for the Markov chain
  is obtained by drawing a sample from $r$.
\end{enumerate}
An important result corresponding to the above independence chain is obtained.
\begin{proposition}\label{prop:uniformergod} Consider an independence
  Metropolis-Hastings algorithm with target density
  $\pi(\tau_h,\tau_c,\Theta)$ and proposal density
  $r(\tau_h,\tau_c,\Theta)$. The resulting Markov chain is uniformly
  ergodic.
\end{proposition}
\begin{proof} 
  Follows directly from Proposition 1 and Theorem 2.1 in \cite{Meng:Twee:1996}.
\end{proof}
An important consequence of Proposition 2 is that a Central Limit
Theorem can be shown to hold for the independence Markov chain, and a
consistent estimate of Monte Carlo standard errors is easily obtained
via consistent batch means \citep{Jone:Hara:Caff:Neat:fixe:2006}, as
discussed in Section \ref{subsec:MCMC}. Furthermore, the Markov chain
can be stopped based on the estimated standard errors attaining a
desired threshold \citep{fleg:hara:jone:2008}. Of course, the fact
that the sampler is uniformly ergodic does not, on its own, imply that
it is an efficient sampler in practice. We therefore study the
efficiency of our indepedence Metropolis-Hastings algorithm in a
variety of real data situations in Section \ref{application}.

We note that $r$ also satisfies conditions for the construction of a
perfect tempering sampler \citep{Moll:Nich:1999}, an algorithm that
utilizes simulated tempering to construct a variant of the
Propp-Wilson perfect sampler \citep{Prop:Wils:1996}. This presents an
intriguing possibility for future research since \cite{Moll:Nich:1999}
report an increase in efficiency by using perfect tempering. However,
as also seen in an in depth study of perfect tempering for such models
in \cite{blinded:2003}, it is challenging to construct a perfect
tempering algorithm that is consistently more efficient than the much
simpler rejection sampler for the examples considered here.
\section{Applications to Cancer and Infant Mortality Data}\label{application}

\subsection{Description of data sets}

We consider a total of three data sets: two on cancer in the U.S.
state of Minnesota, and one on infant mortalities in the United
States. The first two Minnesota cancer data sets were obtained from
the Minnesota Cancer Surveillance System (MCSS), a cancer registry
maintained by the Minnesota Department of Health. The MCSS is
population-based for the state of Minnesota, and collects information
on geographic location and stage at detection for colorectal,
prostate, lung, and female breast cancers.  We illustrate our
computational approaches by analyzing the MCSS data for two of the
cancers, breast and colorectal.  Each of the 87 counties in the data
set has associated with it the total number of cancer cases recorded
between 1995 and 1997, and the number of these detected late. We then
take the expected number of late detections for that county as the
number of cancer cases for that county multiplied by the statewide
rate of late detections. The question of interest is whether there are
clusters of counties in the state of Minnesota with much higher than
expected late detection rates for either cancer. The spatial model
provides smoothed estimates of the relative risk of cancer cases being
detected late in each county. For these data, the posterior
distribution based on the model in Section \ref{sec:model} has 176
dimensions.

We also consider a larger data set on infant mortalities in the United
States.  These data are derived from the Bureau of Health Professions
Area Resource File, which is a county-level database for health
analysis.  Total number of births and deaths are obtained from 1998 to
2000. The geographic units used in this study include all counties of
the following five states: Alabama, Georgia, Mississippi, North
Carolina and South Carolina \citep{healthres:2003}.
The substantive problem of interest with this data is to determine
spatial trends in infant mortality, and finding clusters of regions
with unusually elevated levels in order to study possible
socio-economic contributing factors. This data set, resulting in a
posterior distribution of 910 dimensions, affords an opportunity to
study the performance of our algorithms in a different and potentially
more challenging high dimensional setting.

\subsection{Setting up the algorithm}
To find appropriate heavy-tailed distributions, it is useful to look
at profiles of the approximate marginal posterior distributions (on
log-scale) and find t-distributions that match them reasonably well.
The log-t distributions are obtained by exponentiating these
t-distributions. Once the log-t distributions are specified, the joint
proposal distribution is obtained easily according to
\eqref{eq:SMCMCpost}, and both rejection and independence
Metropolis-Hastings samplers can be constructed. For the rejection
sampler, we find an upper bound $B$ that satisfies \eqref{eq:rejcond}
by numerically optimizing the ratio of the target and candidate
densities. This worked well for our examples, though an alternative
would be to use the ``empirical-sup'' rejection sampler
\citep{Caff:Boot:Davi:2002} to allow for the value of $B$ to be
adaptively estimated by the maximum based on the simulated candidates.
We note that an advantage of the independence Metropolis-Hastings
algorithm is the fact that this upper bound $B$ is not needed.

The rejection sampler based on the heavy-tailed approximation produces
a small number of accepted samples relative to the number of samples
proposed. Similarly, the independence M-H algorithm may have low
acceptance rates, especially for high dimensional distributions. It is
therefore  important to use efficient computational methods
both for drawing the proposals, and for evaluating the
Metropolis-Hastings ratios.  The steps that take up most of the
computational time involve operations on large matrices. 
Thus, for these samplers to be useful in practice it is necessary to
exploit the sparseness of these matrices. We follow \citet{Rue:2001}
by minimizing the bandwidth of the matrices involved and running fast
band-matrix algorithms to speed up matrix computation.  We find that
this dramatically speeds up computation time and makes rejection
sampling practical in some cases, and makes independence
Metropolis-Hastings efficient and practical in all the examples we
consider.  The computer code was entirely implemented in {\tt R}
\citep{Ihak:Gent:1996}, utilizing appropriate sparse matrix routines
written in {\tt FORTRAN} from the well established online resource
Netlib ({\tt
  http://www.netlib.org}). 

\subsection{Results}
The efficiency of the algorithms were compared both in terms of the
number of samples required for the estimates to attain a desired level
of accuracy, as well as the total time taken by the algorithm before
stopping. The same desired level of accuracy was specified for each of
the parameters in both algorithms.  For example, for the breast cancer
data set standard errors for each of the random effects was set to be
no more than 0.01, while it was set to be no more than 2 for each of
the precision parameters. These results are summarized in Table 1.
\begin{table}
  \caption{\label{TABspatperf2} Comparison of rejection sampler and independence MH (I-MH) algorithms}
  \centering
    \begin{tabular}{|c|c|c|c|c|}
      \hline
      data set
      & \multicolumn{2}{c|} {samples required before stopping} & \multicolumn{2}{c|} {total time taken (seconds)}\\
      \cline{2-5}
      & rejection & I-MH & rejection & I-MH\\
      \hline             
      breast cancer        & 4,118   & 29,241   &  2,663s    & 183s\\
      colo-rectal cancer   & 4,735   & 27,225   &  543s      & 170s\\
      infant mortality     & ---     & 97,721   & ---        & 10,066s \\
      \hline
    \end{tabular} 
\end{table}
Clearly, the independence M-H sampler is most time efficient for all
three data sets. The efficiency of the independence M-H algorithm when
compared to rejection sampling is in accordance with results in
\citet{Liu:metr:1996}. For the breast cancer and colo-rectal cancer
data sets, exact sampling is viable. Note that while the number of
samples required is similar for the breast cancer and colo-rectal
cancer data sets, the time taken to produce samples is much longer for
the breast cancer data. This is because fewer exact samples were
returned per unit time for the breast cancer data set.  For the infant
mortality data set, the rejection sampler was unable to attain the
desired accuracy level, even after the algorithm was run for well over
60 hours.  The independence sampler, however, was able to provide
estimates within 3 hours. Hence, the independence sampler may be
practical even when the exact sampler is not.  We feel it is important
to note that the timings above assume that we do not utilize any form
of parallel computing. Both the rejection sampler and the independence
M-H sampler are `embarrassingly' parallel in that each of the draws
from the proposal as well as the target distribution and proposal
evaluations can be done entirely in parallel. 
Hence, with relatively little effort the computational effort can be
reduced by a factor corresponding to the number of available
processors.  Hence, for instance, a 3 hour independence M-H algorithm
would take well under 5 minutes if 50 processors are available. This
is of particular interest since computer clusters with large numbers
of available processors are becoming increasingly accessible for
scientific computing.

A motivation behind exploring rejection sampling is the simplicity and
rigor of sample-based inference: easily computable consistent
estimates of standard errors, avoidance of issues about determining
starting values and not having to rely on ad-hoc approaches for
determining Markov chain length. While the classic rejection sampler
we have constructed based on our heavy-tailed approximation works
surprisingly well in some cases, it is impractical for challenging,
high dimensional examples. We find our independence
Metropolis-Hastings sampler is still efficient and feasible in such
cases, while retaining much of the simplicity of exact sampling. In
particular, to obtain starting values for our independence chain, we
simply simulate from our heavy-tailed approximation $r$. From
Proposition 2, this chain is uniformly ergodic, and we can easily see
from results described in Section \ref{subsec:MCMC} that Monte Carlo
standard errors for different parameters can be estimated
consistently, and these standard errors can be used to provide a
stopping rule for this sampler based on sound principles, just as one
would for exact samplers.  Hence, from the user's perspective,
sample-based inference is no more complicated than for inference based
on the exact samplers.  For the data sets and models considered here,
this sampler resolves all MCMC issues originally raised in Section
\ref{sec:intro}.

\section{Discussion}\label{discussion}
We have demonstrated that it is feasible to implement samplers for
realistic hierarchical Bayesian models in a manner that permits
rigorous estimation of standard errors, while avoiding the usual
issues regarding the determination of simulation lengths. 
We focused on hierarchical models where we are able to construct
accurate heavy-tailed analytical approximations. We described a class
of models to which we can apply our approximation techniques and
derive theoretical results. Our general approximation strategy (in
Section \ref{genmargdist}) is more broadly applicable, as are the
central ideas behind automating the starting and stopping of the
sampler in rigorous fashion. We do not claim that our approach to
constructing the approximation is the best for any given problem
since, for example, other approaches may improve upon the proposals
constructed for the example outlined in Section \ref{proposal}.
Better approximations will naturally lead to more efficient
samplers. Rather, we believe the purpose of this paper is to show that
using carefully derived heavy-tailed proposals along with recent
developments in MCMC theory, sample-based inference can be carried out
in a simple, fairly automated and rigorous fashion for some models.

For the examples in this paper we have explored exact sampling via the
rejection sampler and have successfully demonstrated the applicability
of our methods to some real data sets.  These exact samplers, however,
are generally not feasible for higher dimensional problems.  A fast
mixing independence Metropolis-Hastings algorithm is a more efficient
alternative, while retaining the rigor and simplicity of inference
with exact samplers by guaranteeing that a Central Limit Theorem
holds, and allowing for simple but consistent estimates of standard
errors and intuitive and theoretically justified stopping rules.  The
increased efficiency of the independence chain becomes critical in
higher dimensional problems, where it is viable even when the exact
samplers are not. A potential weakness of our independence chain
algorithm is that it involves block updates of high dimensions, which
can lead to low acceptance rates as the dimensions of the problem
increases.  While this is certainly of concern for very large
dimensional problems, we believe that for moderately high dimensional
problems (thousands of dimensions), a combination of sparse matrix
approaches and the latest high speed parallel computing, holds much
promise for generating and evaluating high dimensional proposals
extremely quickly. Such approaches are being explored elsewhere.

Using a combination of analytical work and modern computing power, our
results suggest that it may be possible to construct rigorous, nearly
automated approaches to sample-based inference for some classes of
models that are of practical importance.  In situations where
multimodality may be an issue, the heavy-tailed approximation also
provides a simple and reasonable approach for generating starting
values for simulating multiple chains on several different processors,
since the approximation is genuinely over-dispersed with respect to
the target distribution. While obtaining theoretical results for each
new model using our approach may prove to be non-trivial in general,
we believe that the methodology outlined for constructing accurate
heavy-tailed approximations may be generally useful, and our fixed
width stopping rules may still be valuable in cases where analytical
results are hard to obtain.


\begin{appendix}
\section{Approximate marginal and conditional distributions}\label{app:derivations}
The full joint posterior distribution from Section \ref{sec:model} is: 
\begin{equation*}
\begin{split}
  \pi(\btheta, \bphi, \tau_h,\tau_c) \propto & \exp\left(\sum_{i=1}^N((\theta_i+\phi_i)Y_i - E_ie^{\theta_i+\phi_i})\right)\exp\left(-\frac{1}{2}\bthet^T(\tau_hI)\bthet\right)\\
  \times &
  \tau_h^{N/2+\alpha_h-1}\tau_c^{M/2+\alpha_c-1}\exp\left(-\frac{1}{2}\bphi^T(\tau_cQ)\bphi\right)\exp{\left(-\frac{\tau_h}{\beta_h} -\frac{\tau_c}{\beta_c}\right)}.
\end{split}
\end{equation*}
The approximate joint posterior distribution based on
(\ref{eq:Approxlikelihood2}), (\ref{eq:delta2}) is:
\begin{equation*}
  \begin{split}
    \hat{\pi}(\Theta,\tau_h,\tau_c|\bY) \propto & \exp\left(-\frac{1}{2}{(\bmuhat-(\bthet + \bphi))}^T V^{-1}{(\bmuhat-(\bthet + \bphi))}-\frac{1}{2}\bthet^T(\tau_hI)\bthet-\frac{1}{2}\bphi^T(\tau_cQ)\bphi\right)\\\
     \times & \tau_h^{N/2+\alpha_h-1}\tau_c^{M/2+\alpha_c-1}\exp\left(-\tau_h/\beta_h-\tau_c/\beta_c\right)
  \end{split}
\end{equation*}
where $\bmuhat=(\log(Y_1/E_1),\dots,\log(Y_N/E_N))^T, M=N-1$,
$\Vinv=\Diag(Y_1,\dots,Y_N)$, $I$ is an N$\times$ N identity matrix, $Q$
is the adjacency matrix described in Section \ref{sec:model}, and
$\alpha_h,\alpha_c,\beta_h,\beta_c>0$ are hyperparameters of the Gamma
density (as in Section \ref{subsec:margdistr}).  To obtain the approximate
marginal posterior distribution $s_1(\tau_h,\tau_c|\bY)$ up to a
constant of proportionality, we integrate
$\hat{\pi}(\Theta,\tau_h,\tau_c|\bY)$ with respect to $\Theta$ to
obtain
\begin{equation*}
  \begin{split}
    s_1(\tau_h,\tau_c|\bY) =& \tau_h^{N/2+\alpha_h-1}\tau_c^{M/2+\alpha_c-1}\exp{(-\tau_h/\beta_h -\tau_c/\beta_c)}\\
    \times & {(\det(\tau_h V^{-1}+\tau_c \Vinv Q + \tau_h \tau_c Q))}^{-1/2}\\
    \times & \exp\left\{-\frac{1}{2}\left(\hat{\Theta}^T C
        \hat{\Theta} + D\hat{\Theta} + k\right)\right\}
  \end{split}
\end{equation*}
with
\begin{equation*}
  C_{2N\times 2N} =
  \begin{bmatrix}
    \Vinv + \tau_h I & +\Vinv\\
    +\Vinv & \Vinv + \tau_c Q
  \end{bmatrix},
\end{equation*}
and
\begin{equation*}
  D_{1 \times 2N}= (-2\bmuhat^TV^{-1}, -2\bmuhat^TV^{-1}),
\end{equation*}
\begin{equation*}
\Thetahat^T = (\hat{\bthet}^T,\hat{\bphi}^T),
\end{equation*}
where $\bthethat=\frac{\tau_c}{\tau_h} Q (I + \frac{\tau_c}{\tau_h} Q + \tau_c V Q)^{-1} \bmuhat$ and $\bphihat=(I + \frac{\tau_c}{\tau_h} Q + \tau_c V Q)^{-1} \bmuhat$. For convenience, we have denoted $C$ by $C(\tau_h,\tau_c)$.
Our heavy-tailed approximation for the joint marginal of $(\tau_h,\tau_c)$ is a product of log-t distributions,
\begin{equation*}
  r_1(\tau_h,\tau_c|\bY) \propto \frac{1}{\tau_h\tau_c} \left[ 1+\frac{1}{\nu_h}\left(\frac{\log(\tau_h) - \mu_h}{\sigma_h}\right)^2\right]^{-(\nu_h+1)/2}\left[ 1+\frac{1}{\nu_c}\left(\frac{\log(\tau_c) - \mu_c}{\sigma_c}\right)^2\right]^{-(\nu_c+1)/2}
\end{equation*}
where $\mu_h,\mu_c \in \Real$, $\sigma_h,\sigma_c > 0$ and
$\nu_h,\nu_c \in \integerZ^+$ are tuned to match the approximate
marginal posterior density $s_1(\tau_h,\tau_c|\bY)$. Our approximation
for the conditional distribution of the random effects parameters,
$\pi(\Theta\mid \bY)$ is a multivariate-t version of (\ref{eq:SMCMCpost}):
\begin{equation*}
r_2(\Theta | \tau_h,\tau_c,\bY) = \frac{|C|^{0.5}\Gamma\left(\frac{\nu_r+2N}{2}\right)}{(\nu_r \pi)^{N/2}\Gamma(\nu_r /2)}
  \left(1 +\frac{1}{\nu_r}(\Theta-\mu_N)^T C (\Theta-\mu_N)\right)^{-(\nu_r+2N)/2}
\end{equation*}
where $\mu_N=- C^{-1}D^T/2$ and $\nu_r \in \integerZ^+$. Note that for
technical reasons (see Lemma \ref{lemma:second}) $Q$ in the matrix $C$
is replaced by a positive definite $\widetilde{Q}$ matrix. Our
heavy-tailed approximation for the joint density
$\pi(\Theta,\tau_h,\tau_c|\bY)$ is
$$r(\Theta,\tau_h,\tau_c|\bY)=r_1(\tau_h,\tau_c|\bY)r_2(\Theta|\tau_h,\tau_c,\bY).$$

\section{Proofs}\label{app:proof}
We begin with several lemmas that will be helpful for proving
Proposition \ref{prop:envelope}. We follow the notation used in
Sections \ref{sec:model} and \ref{proposal}. 
\begin{lemma}\label{lemma:first}
  \begin{equation*}
  \frac{\tau_h^{N/2}\tau_c^{M/2}}{|C(\tau_h,\tau_c)|^{1/2}}
  \le  \left(\frac{\min \bY + \tau_h}{\min \bY}\right)^{1/2}
  \frac{1}{\tau_h^{1/2} \prod_{i=1}^{N-1} \lambda_i^{1/2}},
\end{equation*}
 where $\lambda_1,\dots,\lambda_{N-1}$ are the non-zero eigen-values of $Q$.
\end{lemma}
\begin{proof}
\begin{equation}\label{eq:lemma1a}
  |C| =
  \left|
    \begin{bmatrix}
      V^{-1}+\tau_h I & V^{-1}\\
      V^{-1} & V^{-1} + \tau_c Q
    \end{bmatrix}\right|
    = \left|
    \begin{bmatrix}
      V^{-1}+\tau_h I & V^{-1}\\
      0 & V^{-1} + \tau_c Q - V^{-1}(V^{-1} + \tau_h I)^{-1}V^{-1}
    \end{bmatrix}\right|,
\end{equation}
by subtracting a matrix multiple of the first row block from the
second.  The determinant of the right hand side is the product of the
determinants of the two diagonal blocks.  The determinant of the first
block is bounded by
\begin{equation}\label{eq:lemma1b}
  \left|V^{-1} + \tau_h I\right| \ge \tau_h^N
\end{equation}
The second diagonal block is
\begin{equation*}
  V^{-1} + \tau_c Q - V^{-1}(V^{-1} + \tau_h)^{-1}V^{-1}
  = \diag(\bY-\bY^2/(\bY+\tau_h)) + \tau_c Q
  = \diag\left(\frac{\bY \tau_h}{\bY + \tau_h}\right) + \tau_c Q
\end{equation*}
The determinant can be bounded below by replacing the $Y_i$ by their
minimum value ($\min \bY$).  We can then write $Q = U \Lambda U^T$
with $\Lambda = \diag(\lambda_1,\dots,
\lambda_N)$.
This gives the lower bound
\begin{align}\label{eq:lemma1c}
  |V^{-1} + \tau_c Q - V^{-1}(V^{-1} + \tau_h)^{-1}V^{-1}| &\ge
  \prod_{i=1}^N \left(\frac{\min \bY \tau_h}{\min \bY + \tau_h} + \tau_c \lambda_i\right)\\
  &\ge \left(\frac{\min \bY}{\min \bY + \tau_h}\right)
  \tau_h \tau_c^M \prod_{i=1}^{N-1} \lambda_i .
\end{align}
 Combining \eqref{eq:lemma1a},\eqref{eq:lemma1b},\eqref{eq:lemma1c}, we obtain the bound, 
\begin{equation*}
  \frac{\tau_h^{N/2}\tau_c^{M/2}}{|C(\tau_h,\tau_c)|^{1/2}}
  \le  \left(\frac{\min \bY + \tau_h}{\min \bY}\right)^{1/2}
  \frac{1}{\tau_h^{1/2} \prod_{i=1}^{N-1} \lambda_i^{1/2}}
\end{equation*}
\end{proof}
\begin{lemma}\label{lemma:second}
Let the center for the multivariate-$t$ approximation, $r_2$,  be 
\begin{equation*}
  \mu_N = 
  \begin{bmatrix}
    \mu_{N,\theta}\\
    \mu_{N,\phi}
  \end{bmatrix} = 2
  \begin{bmatrix}
    V^{-1} + \tau_h I & V^{-1}\\
    V^{-1} & V^{-1} + \tau_c Q
  \end{bmatrix}
  \begin{bmatrix}
    V^{-1} \hat{\mu}\\
    V^{-1} \hat{\mu}
  \end{bmatrix}.
\end{equation*}
Then, the quadratic terms $\|\tau_c \mu_{N,\phi}\|$ and $\|\tau_h
  \mu_{N,\theta}\|$  are bounded,
\begin{equation*}
\begin{split}
  \|\tau_c \mu_{N,\phi}\| & \le 2 \|Q^{-1}\| \, \|\diag(\bY^2)\| \, \|\hat{\mu}\|\\
  \|\tau_h \mu_{N,\theta}\| & \le 2 \|\hat{\mu}\|.
\end{split}
\end{equation*}
\end{lemma}
\begin{proof}
  We will use a non-singular version of the $Q$ matrix. Let $
  \widetilde{Q} = Q + \delta I$, for a small $\delta > 0$ in computing
  $\mu_N$. This does not change the model, and is necessary only for
  the center, not the spread of the approximation. 
Using a partitioned form of the inverse, we can write
\begin{equation*}
  C^{-1} =
  \begin{bmatrix}
    A & -A V^{-1} (V^{-1} + \tau_c \widetilde{Q})^{-1}\\ 
    - B V^{-1} (V^{-1} + \tau_h I)^{-1} & B
  \end{bmatrix}
\end{equation*}
with
\begin{align*}
  A &= [V^{-1} + \tau_h I - V^{-1} (V^{-1} + \tau_c \widetilde{Q})^{-1} V^{-1}]^{-1}\\
  B &= [\tau_c \widetilde{Q} + V^{-1} - V^{-1}(V^{-1} + \tau_h I)^{-1}V^{-1}]^{-1}
\end{align*}
Since $V^{-1} = \diag(\bY)$, $B$ simplifies to
\begin{equation*}
  B =
  \left[\tau_c \widetilde{Q} + \diag\left(\bY - \frac{\bY^2}{\bY + \tau_h}\right)\right]^{-1}
  = \left[\tau_c \widetilde{Q} + \diag\left(\frac{\tau_h \bY}{\bY + \tau_h}\right)\right]^{-1}
\end{equation*}
and
\begin{equation*}
  C^{-1} =
  \begin{bmatrix}
    A & -A V^{-1} (V^{-1} + \tau_c \widetilde{Q})^{-1}\\ 
    - B\, \diag\left(\frac{\bY}{\bY + \tau_h}\right) & B
  \end{bmatrix}  
\end{equation*}
Now
\begin{equation*}
  \mu_{N,\phi} = 2 B \left(I-\diag\left(\frac{\bY}{\bY + \tau_h}\right)\right)
  \diag(\bY)\hat{\mu}
  = B \, \diag\left(\frac{\tau_h \bY^2}{\bY + \tau_h}\right)\hat{\mu}
\end{equation*}
To see that $\|\mu_{N,\phi}\|$ is bounded, note that
\begin{equation*}
  B^{-1} \ge \diag\left(\frac{\tau_h \bY}{\bY + \tau_h}\right)
\end{equation*}
in non-negative definite matrix order and therefore
\begin{equation*}
  \|\mu_{N,\phi}\| \le
  2 \left\|\diag\left(\frac{\tau_h \bY}{\bY + \tau_h}\right)^{-1}
    \diag\left(\frac{\tau_h \bY^2}{\bY + \tau_h}\right) \hat{\mu}\right\|
  = 2 \|\diag(\bY)\hat{\mu}\|
\end{equation*}
To bound $\|\tau_c \mu_{N,\phi}\|$, we can similarly use
\begin{equation*}
  B^{-1} \ge \tau_c \widetilde{Q}
\end{equation*}
and
\begin{equation*}
  \diag\left(\frac{\tau_h \bY^2}{\bY + \tau_h}\right) \le \diag(\bY^2)
\end{equation*}
to conclude that
\begin{equation*}
  \|\tau_c \mu_{N,\phi}\| \le 2 \|\widetilde{Q}^{-1}\| \, \|\diag(\bY^2)\| \, \|\hat{\mu}\|.
\end{equation*}
The component $\mu_{N,\theta}$ is given by
\begin{equation*}
  \mu_{N,\theta} = 2 A [I - V^{-1}(V^{-1} + \tau_c Q)^{-1}] V^{-1} \hat{\mu}
  = 2 A [V^{-1} - V^{-1}(V^{-1} + \tau_c Q)^{-1} V^{-1}] \hat{\mu}
\end{equation*}
To bound $\|\mu_{N\theta}\|$ and $\|\tau_h \mu_{N\theta}\|$ we can use
the inequalities
\begin{equation*}
  A \le [V^{-1} V^{-1}(V^{-1} + \tau_c Q)^{-1} V^{-1}]^{-1}
\end{equation*}
and
\begin{equation*}
  A \le \frac{1}{\tau_c} I
\end{equation*}
Using the first inequality
\begin{equation*}
  \|\mu_{N,\theta}\| \le
  2 \|[V^{-1} V^{-1}(V^{-1} + \tau_c Q)^{-1} V^{-1}]^{-1}
  [V^{-1} - V^{-1}(V^{-1} + \tau_c Q)^{-1} V^{-1}] \hat{\mu}\|
  = 2 \|\hat{\mu}\|
\end{equation*}
and using the second inequality
\begin{equation*}
  \|\tau_h \mu_{N,\theta}\| \le 2 \tau_h \left\|\frac{1}{\tau_h}\hat{\mu}\right\|
  = 2 \|\hat{\mu}\|.
\end{equation*}
\end{proof}
\begin{lemma}\label{lemma:third}
\begin{equation*}
  \exp\left(\sum_{i=1}^N((\theta_i+\phi_i)Y_i - E_i e^{\theta_i+\phi_i}
  -\frac{\tau_h}{2}\theta^T\theta
  -\frac{\tau_c}{2}\phi^T Q \phi\right)
  \left(1 + \frac{1}{\nu_r}(\Theta - \mu_N) C (\Theta - \mu_N)\right)^{(\nu_r+2N)/2}
\end{equation*}
is bounded. 
\end{lemma}
\begin{proof}
First note that for $A, B \ge 0$ we have
\begin{equation*}
  1+A+B = (1+A)\left(1+\frac{B}{1+A}\right) \le (1+A)(1+B)
\end{equation*}
\noindent It is therefore sufficient to bound each of the three terms
\begin{align}
  \exp\left(\sum_{i=1}^N((\theta_i+\phi_i)Y_i - E_i
    e^{\theta_i+\phi_i)}\right)
  \left(1 + \frac{1}{\nu_r}(\Theta - \mu_N) C^{*} (\Theta - \mu_N)\right)^{(\nu_r+2N)/2}\label{eq:lemma3a}\\
  \exp\left(-\frac{\tau_h}{2}\theta^T\theta\right)
  \left(1 + \frac{1}{\nu_r}(\Theta - \mu_N) C^{**} (\Theta - \mu_N)\right)^{(\nu_r+2N)/2}\label{eq:lemma3b}\\
    \exp\left(-\frac{\tau_c}{2}\phi^T Q \phi\right) \left(1 +
      \frac{1}{\nu_r}(\Theta - \mu_N) C^{***} (\Theta -
      \mu_N)\right)^{(\nu_r+2N)/2}\label{eq:lemma3c}
\end{align}
where
\begin{align*}
  C^{*} &=
  \begin{bmatrix}
    V^{-1} & V^{-1}\\
    V^{-1} & V^{-1}\\
  \end{bmatrix} &
  C^{**} &=
  \begin{bmatrix}
    \tau_h I & 0\\
    0 & 0
  \end{bmatrix} &
  C^{***} &=
  \begin{bmatrix}
    0 & 0\\
    0 & \tau_c Q
  \end{bmatrix}
\end{align*}

To bound (\ref{eq:lemma3a}), reparameterize in terms of $w_i =
\theta_i+\phi_i$ and $v_i = \theta_i - \phi_i$. Then the quadratic
form in (\ref{eq:lemma3a}) can be written as
\begin{align*}
  (\Theta - \mu_N)^T C^{*} (\Theta - \mu_N) &=
  [(w - \mu_w^{*})^T\quad (v - \mu_v^{*})^T]
  \begin{bmatrix}
    \frac{1}{2} &  \frac{1}{2}\\
    \frac{1}{2} & -\frac{1}{2}
  \end{bmatrix}
  \begin{bmatrix}
    V^{-1} &     V^{-1}\\
    V^{-1} &     V^{-1}
  \end{bmatrix}
  \begin{bmatrix}
    \frac{1}{2} &  \frac{1}{2}\\
    \frac{1}{2} & -\frac{1}{2}
  \end{bmatrix}
  \begin{bmatrix}
    w - \mu_w^{*}\\
    v - \mu_v^{*}
  \end{bmatrix}\\
  &= (w - \mu_w^{*})^T V^{-1} (w - \mu_w^{*})\\
  &= \sum_{i=1}^N Y_i (w_i - \mu_i^{*})^2
\end{align*}
 with
\begin{equation*}
  \mu^{*} =
  \begin{bmatrix}
    \mu_w^{*}\\
    \mu_v^{*}
  \end{bmatrix} = 
  \begin{bmatrix}
    \frac{1}{2}I &     \frac{1}{2}I \\
    \frac{1}{2}I &    -\frac{1}{2}I 
  \end{bmatrix} \mu_N
\end{equation*}
This does not depend on $v$. So the quadratic form in (\ref{eq:lemma3a}) is bounded
by
\begin{equation*}
  (\Theta - \mu_N)^T C^{*} (\Theta - \mu_N)
  \le K_1 + K_2 \sum_{i=1}^NY_i w_i^2
\end{equation*}
for some constants $K_1$ and $K_2$. The term (\ref{eq:lemma3a}) can thus be bounded by
\begin{equation*}
  (1 + K_1)^{(\nu_r+2N)/2} \prod_{i=1}^N \left[
    \exp(w_i Y_i - E_i e^{w_i})\left(1 + \frac{K_2}{\nu_r}Y_i w_i^2\right)^{(\nu_r+2N)/2}
    \right]
\end{equation*}
None of these terms involves $\tau_h$ or $\tau_c$ and since $Y_i > 0$
each term
\begin{equation*}
  \exp(w_i Y_i - E_i e^{w_i})\left(1 + \frac{K_2}{\nu_r}Y_i w_i^2\right)^{(\nu_r+2N)/2}
\end{equation*}
is bounded in $w_i$.  So (\ref{eq:lemma3a}) is bounded. Using Lemma
\ref{lemma:second}, the quadratic forms in (\ref{eq:lemma3b}) and
(\ref{eq:lemma3c}) are bounded by
\begin{equation*}
    (\Theta - \mu_N)^T C^{**} (\Theta - \mu_N)
    = \tau_h(\theta - \mu_{N,\theta})^T (\theta - \mu_{N,\theta})
    \le K_3 + K_4 \tau_h\theta^T\theta
\end{equation*}
and
\begin{equation*}
    (\Theta - \mu_N)^T C^{***} (\Theta - \mu_N)
    = \tau_c(\phi - \mu_{N,\phi})^T Q (\phi - \mu_{N,\phi})
    \le K_5 + K_6 \tau_c\phi^T Q \phi
\end{equation*}
for some constants $K_3, K_4, K_5$, and $K_6$.  Thus term (\ref{eq:lemma3b}) is
bounded by
\begin{equation*}
  (1+K_3)^{(\nu_r+2N)/2}
  \sup_{z \ge 0}\left\{
    \exp\left(-\frac{z}{2}\right)
    \left(1 + \frac{K_4}{\nu_r}z\right)^{(\nu_r+2N)/2}\right\} < \infty
\end{equation*}
and term (\ref{eq:lemma3c}) is bounded analogously. 
\end{proof}
We can now utilize the above lemmas to prove Proposition \ref{prop:envelope}.
\begin{proof}[Proposition \ref{prop:envelope}] The ratio
  $\pi(\tau_h,\tau_c,\Theta)/r(\tau_h,\tau_c,\Theta)$ can be written
  as
\begin{equation*}
  \begin{split}
    & \exp\left(\sum_{i=1}^N((\theta_i+\phi_i)Y_i - E_i e^{\theta_i+\phi_i}
      -\frac{\tau_h}{2}\theta^T\theta
      -\frac{\tau_c}{2}\phi^T Q \phi\right)
    \left(1 + \frac{1}{\nu_r}(\Theta - \mu_N) C (\Theta - \mu_N)\right)^{(\nu_r+2N)/2}\\
    \times &
    \exp{\left(-\frac{\tau_h}{\beta_h} -\frac{\tau_c}{\beta_c}\right)}\left[ 1+\frac{1}{\nu_h}\left(\frac{\log(\tau_h) - \mu_h}{\sigma_h}\right)^2\right]^{(\nu_h+1)/2} \left[ 1+\frac{1}{\nu_c}\left(\frac{\log(\tau_c) - \mu_c}{\sigma_c}\right)^2\right]^{(\nu_c+1)/2} \\
    \times & \tau_h^{N/2+\alpha_h}\tau_c^{M/2+\alpha_c} \det(C)^{-0.5}.
  \end{split}
\end{equation*}
The first term is bounded by Lemma \ref{lemma:third}.  We can apply
Lemma \ref{lemma:first} to bound the determinant in the third term. Hence, it
 follows that the above ratio (ignoring constants) is bounded by
\begin{equation*}
  \begin{split}
    & \exp{\left(-\frac{\tau_h}{\beta_h} -\frac{\tau_c}{\beta_c}\right)}\left[ 1+\frac{1}{\nu_h}\left(\frac{\log(\tau_h) - \mu_h}{\sigma_h}\right)^2\right]^{(\nu_h+1)/2} \left[ 1+\frac{1}{\nu_c}\left(\frac{\log(\tau_c) - \mu_c}{\sigma_c}\right)^2\right]^{(\nu_c+1)/2} \\
    \times &  \tau_h^{\alpha_h} \tau_c^{\alpha_c} \left(\frac{\min \bY + \tau_h}{\min \bY}\right)^{1/2}  \frac{1}{\tau_h^{1/2}},
 \end{split}
\end{equation*}
which is bounded as long as $\alpha_h\geq 1$.
\end{proof}
\end{appendix}

\bibliographystyle{apalike} 
\bibliography{htspatperf}

\end{document}